\newcommand{\longversion}[1]{#1}
\newcommand{\shortversion}[1]{}

\longversion{\documentclass[11pt]{article}
\author{{\bf Palash Dey}\\Indian Institute of Science, Bangalore\\palash@csa.iisc.ernet.in
\and
{\bf Neeldhara Misra}\\Indian Institute of Technology, Gandhinagar\\mail@neeldhara.com
}}

\shortversion{
\documentclass{article}
\usepackage{ijcai16}

\usepackage{times}

\date{}

\author{{\bf Palash Dey}\\Indian Institute of Science, Bangalore\\palash@csa.iisc.ernet.in
\and
{\bf Neeldhara Misra}\\Indian Institute of Technology, Gandhinagar\\mail@neeldhara.com
}

}

\usepackage{amssymb,amsmath}
\usepackage[usenames,dvipsnames]{color}
\usepackage[colorlinks=true,linkcolor=blue,citecolor=Mahogany]{hyperref}
\usepackage{makecell}
\usepackage{multirow}
\usepackage{mathtools}

\longversion{\usepackage{fullpage}}

\newenvironment{proof}{\noindent{\em Proof:}}{ \hfill $\square$\\ }

\usepackage[usenames,svgnames,dvipsnames]{xcolor}

\usepackage{titling}

\usepackage{url}

\usepackage{algpseudocode,algorithm,algorithmicx}

\algrenewcommand\algorithmicrequire{\textbf{Input:}}
\algrenewcommand\algorithmicensure{\textbf{Output:}}
\algnewcommand{\Initialize}[1]{%
  \State \textbf{Initialize:}
  \Statex \hspace*{\algorithmicindent}\parbox[t]{.8\linewidth}{\raggedright #1}
}
\algdef{SE}[FUNCTION]{Function}{EndFunction}%
   [2]{\algorithmicfunction\ \textproc{#1}\ifthenelse{\equal{#2}{}}{}{(#2)}}%
   {\algorithmicend\ \algorithmicfunction}%

\newcommand{\el}{\ensuremath{\ell}\xspace}

\newcommand{\suc}{\ensuremath{\succ} \xspace}
\renewcommand{\ge}{\geqslant}
\renewcommand{\le}{\leqslant}

\shortversion{
\allowdisplaybreaks
\allowbreak

\usepackage{setspace}
\setlength{\topsep}{1pt}
\setlength{\partopsep}{1pt plus 1pt minus 1pt}
\setlength{\parskip}{0pt}
\setlength{\parindent}{5pt}
\usepackage{titlesec}

\titlespacing\section{0pt}{2pt plus 1pt minus 1pt}{2pt plus 1pt minus 1pt}
\titlespacing\subsection{0pt}{2pt plus 1pt minus 1pt}{2pt plus 1pt minus 1pt}
\titlespacing\paragraph{0pt}{2pt plus 1pt minus 1pt}{2pt plus 1pt minus 1pt}
\setlength{\belowdisplayskip}{1pt} 
\setlength{\abovedisplayskip}{1pt} 

\floatsep 1pt plus 2pt minus 2pt
\textfloatsep 1pt plus 2pt minus 1pt
\intextsep 1pt plus 2pt minus 2pt
\dblfloatsep 1pt plus 2pt minus 2pt
\dbltextfloatsep 1pt plus 2pt minus 1pt

\usepackage{etoolbox}
\newcommand{\zerodisplayskips}{%
  \setlength{\abovedisplayskip}{1pt}
  \setlength{\belowdisplayskip}{1pt}
  \setlength{\abovedisplayshortskip}{1pt}
  \setlength{\belowdisplayshortskip}{1pt}}
\appto{\normalsize}{\zerodisplayskips}
\appto{\small}{\zerodisplayskips}
\appto{\footnotesize}{\zerodisplayskips}
\setlength{\textfloatsep}{3ex}
\setlength{\intextsep}{1ex}

\usepackage{subfig}
\captionsetup{belowskip=1ex,aboveskip=1ex}

\usepackage{enumitem}
\setlist{topsep=0.5ex,itemsep=-1ex,partopsep=0.5ex,parsep=1ex,leftmargin=0pt,itemindent=*}

}

\newcommand{\Query}{\textsc{Query}}

\usepackage{MnSymbol,wasysym}

\usepackage{xspace}

\DeclareMathOperator*{\argmin}{arg\!min}

\newcommand{\BigO}{\ensuremath{\mathcal{O}}\xspace}
\newcommand{\true}{\textsc{true}\xspace}
\newcommand{\false}{\textsc{false}\xspace}
\newcommand{\pr}{\ensuremath{\prime}\xspace}

\newcommand{\PE}{\textsc{Preference Elicitation}\xspace}
\newcommand{\CW}{\textsc{Weak Condorcet Winner}\xspace}

\renewcommand{\AA}{\ensuremath{\mathcal A}\xspace}

\newcommand{\CC}{\ensuremath{\mathcal C}\xspace}

\newcommand{\EE}{\ensuremath{\mathcal E}\xspace}

\newcommand{\LL}{\ensuremath{\mathcal L}\xspace}

\newcommand{\OO}{\ensuremath{\mathcal O}\xspace}
\newcommand{\PP}{\ensuremath{\mathcal P}\xspace}
\newcommand{\QQ}{\ensuremath{\mathcal Q}\xspace}
\newcommand{\RR}{\ensuremath{\mathcal R}\xspace}
\renewcommand{\SS}{\ensuremath{\mathcal S}\xspace}
\newcommand{\TT}{\ensuremath{\mathcal T}\xspace}

\newcommand{\VV}{\ensuremath{\mathcal V}\xspace}

\newcommand{\XX}{\ensuremath{\mathcal X}\xspace}
\newcommand{\YY}{\ensuremath{\mathcal Y}\xspace}

\usepackage{nicefrac}

\newcommand{\nfrac}{\nicefrac}

\newtheorem{observation}{\bf Observation}
\newtheorem{theorem}{\bf Theorem}
\newtheorem{lemma}{\bf Lemma}
\newtheorem{corollary}{\bf Corollary}
\newtheorem{definition}{\bf Definition}

\newcommand{\eps}{\varepsilon}
\renewcommand{\epsilon}{\eps}

\newcommand{\ignore}[1]{}
\usepackage[poorman]{cleveref}

\crefname{theorem}{Theorem}{\bf Theorem}
\crefname{observation}{Observation}{\bf Observation}
\crefname{lemma}{Lemma}{\bf Lemma}
\crefname{corollary}{Corollary}{\bf Corollary}
\crefname{proposition}{Proposition}{\bf Proposition}
\crefname{definition}{Definition}{\bf Definition}
\crefname{claim}{Claim}{\bf Claim}

\title{{\bf Elicitation for Preferences Single Peaked on Trees}}

\begin{document}

\maketitle

\begin{abstract}

\longversion{In multiagent systems, we often have a set of agents each of which have a preference ordering over a set of items and one would like to know these preference orderings for various tasks, for example, data analysis, preference aggregation, voting etc. However, we often have a large number of items which makes it impractical to ask the agents for their complete preference ordering. In such scenarios, we usually elicit these agents' preferences by asking (a hopefully small number of) comparison queries --- asking an agent to compare two items.} 
\shortversion{
Eliciting preferences of a set of agents over a set of items is a problem of fundamental interest in artificial intelligence in general and social choice theory in particular.}
Prior works on preference elicitation focus on unrestricted domain and the domain of single peaked preferences and show that the preferences in single peaked domain can be elicited by much less number of queries compared to unrestricted domain. We extend this line of research and study preference elicitation for single peaked preferences on trees which is a strict superset of the domain of single peaked preferences. We show that the query complexity crucially depends on the number of leaves, the path cover number, and the distance from path of the underlying single peaked tree, whereas the other natural parameters like maximum degree, diameter, pathwidth do not play any direct role in determining query complexity. We then investigate the query complexity for finding a weak Condorcet winner for preferences single peaked on a tree and show that this task has much less query complexity than preference elicitation. Here again we observe that the number of leaves in the underlying single peaked tree and the path cover number of the tree influence the query complexity of the problem.
 
\end{abstract}

\section{Introduction}

In multiagent systems, we often have scenarios where agents have to arrive at a consensus when choosing between multiple options. Typically, the agents have preferences over a set of items, and the problem of aggregating these preferences in a suitable manner is one of the most well-studied problems in social choice theory~\cite{brandt2015handbook}. There are many ways of expressing preferences over a set of alternatives. One of the most comprehensive ways is to specify a complete ranking over the set of alternatives. However, one of the downsides of this model is the fact that it can be expensive to solicit a preference when there are a large number of alternatives, and many agents are involved. 

Since asking agents to provide their complete rankings is impractical, a popular notion is one of \textit{elicitation}, where we ask agents simple \textit{comparison queries}, such as if they prefer alternative $X$ over $Y$. This naturally gives rise to the problem of \textit{preference elicitation}, where we hope to recover the complete ranking (or possibly the most relevant part of the ranking) based on a small number of queries. 

The paradigm of \textit{voting} is a popular general methodology for aggregating preferences, where one devises ``voting rules'' for mapping a collection of preferences (which we refer to as votes) to either a \textit{winning alternative} or a \textit{consensus ranking}. Keeping in line with the terminology used in voting, we will refer to alternatives as candidates, and a collection of votes will be termed a preference profile. In the context of a fixed voting rule, we may also want to query the voters up to the point of determining the winner (or the aggregate ranking, as the case may be). Yet another refinement in this setting is when we have prior information about how agents are likely to vote, and we may want to determine which voters to query first, to be able to quickly rule out a large number of alternatives, as explored by~\cite{conitzer2002vote}. 

When our goal is to elicit preferences that have no prior structure, one can demonstrate scenarios where it is imperative to ask each agent (almost) as many queries as would be required to determine an arbitrary ranking. However, in recent times, there has been considerable interest in voting profiles that are endowed with additional structure. The motivation for this is two-fold. The first is that in several application scenarios commonly considered, it is rare that votes are ad-hoc, demonstrating no patterns whatsoever. For example, the notion of \textit{single-peaked preferences}, which we will soon discuss at length, forms the basis of several studies in the analytical political sciences~\cite{hinich1997analytical}. In his work on eliciting preferences that demonstrate the ``single-peaked'' structure,~\cite{Conitzer09} argues that the notion of single-peakedness is also a reasonable restriction in applications outside of the domain of political elections. 

The second motivation for studying restricted preferences is somewhat more technical, but is just as compelling. To understand why structured preferences have received considerable attention from social choice theorists, we must first take a brief detour into some of the foundational ideas that have shaped the landscape of voting rules as we understand them today. As it turns out, the axiomatic approach of social choice involves defining certain ``properties'' that formally capture the quality of a voting rule. For example, we would not want a voting rule to be, informally speaking, a \textit{dictatorship}, which would essentially mean that it discards all but one voter's input. Unfortunately, a series of cornerstone results establish that it is impossible to devise voting rules which respect some of the simplest desirable properties. Indeed, the classic works of Arrow \cite{arrow1950difficulty} and Gibbard-Satterthwaite \cite{gibbard1973manipulation,satterthwaite1975strategy} show that there is no straight forward way to simultaneously deal with properties like voting paradoxes, strategy-proofness, nondictatorship, unanimity etc. We refer to \cite{moulin1991axioms} for a more elaborate discussion. Making the matter worse, many classical voting rules turn out to be computationally intractable. 

This brings us to the second reason for why structured preferences are an important consideration. The notion of single-peakedness that we mentioned earlier is an excellent illustration (we refer the reader to~\Cref{sec:prelim} for the formal definition). Introduced by~\cite{black1948rationale}, it not only captures the essence of structure in political elections, but also turns out to be extremely conducive to many natural theoretical considerations. To begin with, one can devise voting rules that are ``nice'' with respect to several properties, when preferences are single-peaked. Further, they are structurally elegant from the view of winner determination, since they always admit a {\em weak Condorcet winner} --- a candidate which is not defeated by any other candidate in pairwise election --- thus working around the Condorcet paradox which is otherwise a prominent concern in the general scenario. In a landmark contribution, \cite{brandt2015bypassing} show that several computational problems that are intractable in the general setting become polynomially solvable when we consider single-peaked preferences. 

A natural question at this point is if the problem of elicitation becomes any easier --- that is, if we can get away with fewer queries --- by taking advantage of the structure provided by single-peakedness. It turns out that the answer to this is in the affirmative, as shown in a detailed study by~\cite{Conitzer09}. The definition of single-peakedness involves an ordering over the candidates (called the \textit{harmonious ordering} by some authors). The work of~\cite{Conitzer09} shows that $\BigO(mn)$ queries suffice, assuming either that the harmonious ordering is given, or one of the votes is known. Dey and Misra~\cite{deycross} show a query complexity bound of $\BigO(mn)$ for the domain of single crossing profiles and a large number of voters.

We now return to the theme of structural restrictions on preferences. As it turns out, the single peaked preference domain has subsequently been generalized to single peakedness on trees (roughly speaking, these are profiles that are single peaked on every path) \cite{demange1982single,trick1989recognizing}. This is a class that continues to exhibit many desirable properties of single peaked domains. For example, there always exists a weak Condorcet winner and further, many voting rules that are intractable in unrestricted domain are polynomial time computable if the underlying single peaked tree is ``nice'' \cite{yu2013multiwinner,peters2016preferences}. We note the class of profiles that are single peaked on trees are substantially more general than the class of single peaked preferences. Note that the latter is a special case since a path is, in particular, a tree. Our work here addresses the issue of elicitation on profiles that are single-peaked on trees, and can be seen as a significant generalization of the results in~\cite{Conitzer09}. We now detail the specifics of our contributions.



\begin{table*}[htbp]
 \begin{center}
 \renewcommand{\arraystretch}{1.1}
  \begin{tabular}{|c|c|c|}\hline
   Parameter & Upper Bound & Lower Bound \\\hline\hline
   
   Path width ($w$) & $\BigO(mn\log m)$ \Cref{obs:trivial} & $\Omega(mn\log m)$ even for $w = 1, \log m$ \Cref{cor:pwidth_elicit_lb} \\\hline
   
   Maximum degree ($\Delta$) & $\BigO(mn\log m)$ \Cref{obs:trivial} & $\Omega(mn\log m)$ even for $\Delta = 3, m-1$ \Cref{cor:maxdeg_elicit_lb} \\\hline
   
   Path cover number ($k$) & $\BigO(mn\log k)$ \Cref{thm:pcover_elicit_ub} & $\Omega(mn\log k)$ \Cref{cor:pcover_elicit_lb} \\\hline
   
   Number of leaves (\el) & $\BigO(mn\log \el)$ \Cref{cor:leaves_elicit_ub} & $\Omega(mn\log \el)$ \Cref{thm:leaves_elicit_lb} \\\hline
   
   Distance from path ($d$) & $\BigO(mn + nd\log d)$ \Cref{thm:pdist_elicit_ub} & $\Omega(mn + nd\log d)$ \Cref{thm:pdist_elicit_lb} \\\hline
   
   Diameter ($\omega$) & $\BigO(mn\log m)$ \Cref{obs:trivial} & $\Omega(mn\log m)$ even for $\omega = 2, \nfrac{m}{2}$ \Cref{cor:dia_elicit_lb} \\\hline
  \end{tabular}
 \end{center}
\caption{Summary of query complexity bounds for \PE. }\label{tbl:elicit}
\end{table*}

\longversion{\subsection{Motivation}}

\longversion{\subsection{Our Contributions}}
\shortversion{\paragraph*{Our Contributions.}} We study the query complexity for preference elicitation when the preference profile is single peaked on a tree. We provide {\em tight} connections between various parameters of the underlying tree and the query complexity for preference elicitation. Our broad goal is to provide a suitable generalization of preference elicitation for single peaked profiles to profiles that are single peaked on trees. Therefore, we consider various ways of quantifying the ``closeness'' of a tree to a path, and reflect on how these measures might factor into the query complexity of an algorithm that is actively exploiting the underlying tree structure.  

We summarize our results for preference elicitation in \Cref{tbl:elicit}, where the readers will note that most of the parameters (except diameter) chosen are small constants (typically zero, one or two) when the tree under consideration is a path. Observe that in some cases --- such as the number of leaves, or the path cover number --- the dependence on the parameter is transparent (and we recover the results of \cite{Conitzer09} as a special case), while in other cases, it is clear that the perspective provides no additional mileage (the non-trivial results here are the matching lower bounds).

In terms of technique, our strategy is to ``scoop out the paths from the tree'' and use known algorithms to efficiently elicit the preference on the parts of the trees that are paths. We then efficiently merge this information across the board, and that aspect of the algorithm varies depending on the parameter considered. The lower bounds typically come from trees that provide large ``degrees of freedom'' in reordering candidates, typically these are trees that don't have too many long paths (such as stars). The arguments are often subtle but intuitive. 

We then study the query complexity for finding a weak Condorcet winner of a preference profile which is single peaked on a tree. Here, we are able to show that a weak Condorcet winner can be found with far fewer queries than the corresponding elicitation problem. In particular, we establish that a weak Condorcet winner can be found using $\BigO(mn)$ many queries for profiles that are single peaked on trees [\Cref{{thm:cw_gen_ub}}], and we also show that this bound is the best that we can hope for [\Cref{thm:cw_gen_lb}]. We also consider the problem for the special case of single peaked profiles. While~\cite{Conitzer09} showed that $\Omega(mn)$ queries are necessary to determine the \textit{aggregate ranking}, we show that only $\BigO(n\log m)$ queries suffice if we are just interested in (one of the) weak Condorcet winners. Moreover, we show this bound is tight under the condition that the algorithm does not interleave queries to different voters [\Cref{thm:con_sp_lb}] (our algorithm indeed satisfies this condition). 

Finally, expressing the query complexity for determining a weak Condorcet winner in terms of a measure of closeness to a path, we show an algorithm with query complexity $\BigO(nk\log m)$ where $k$ is the path cover number of \TT [\Cref{thm:cw_pc_ub}] or the number of leaves in \TT [\Cref{cor:cw_leaf_ub}].
 \longversion{We now elaborate further on our specific contributions for preference elicitation.

\begin{itemize}
 \item We design novel algorithms for preference elicitation for profiles which are single peaked on a tree with \el leaves with query complexity $\BigO(mn\log \el)$ [\Cref{cor:leaves_elicit_ub}]. Moreover, we prove that there exists a tree \TT with \el leaves such that any preference elicitation algorithm for profiles which are single peaked on tree \TT has query complexity $\Omega(mn\log\el)$ [\Cref{thm:leaves_elicit_lb}]. We show similar results for the parameter path cover number of the underlying tree [\Cref{thm:pcover_elicit_ub,cor:pcover_elicit_lb}]. We provide a preference elicitation algorithm with query complexity $\BigO(mn + nd\log d)$ for single peaked profiles on trees which can be made into a path by deleting at most $d$ nodes [\Cref{thm:pdist_elicit_ub}]. We show that our query complexity upper bound is tight up to constant factors [\Cref{thm:pdist_elicit_lb}]. These results show that the query complexity tightly depends on the number of leaves, the path cover number, and the distance from path of the underlying tree.
 
 \item We then show that there exists a tree \TT with pathwidth one or $\log m$ [\Cref{cor:pwidth_elicit_lb}] or maximum degree is $3$ or $m-1$ [\Cref{cor:maxdeg_elicit_lb}] or  diameter is $2$ or $\nfrac{m}{2}$ [\Cref{cor:dia_elicit_lb}] such that any preference elicitation algorithm for single peaked profiles on the tree \TT has query complexity $\Omega(mn\log m)$. These results show that the query complexity of preference elicitation does not directly depend on the parameters above.
\end{itemize}

We next study query complexity for finding a weak Condorcet winner for profiles which are single peaked on trees and we have the following results.

\begin{itemize}
 \item We show that a weak Condorcet winner can be found using $\BigO(mn)$ many queries for profiles that are single peaked on trees [\Cref{{thm:cw_gen_ub}}] which is better than the query complexity for preference elicitation. Moreover, we prove that this bound is tight in the sense that any algorithm for finding a weak Condorcet winner for profiles that are single peaked on stars has query complexity $\Omega(mn)$ [\Cref{thm:cw_gen_lb}].
 
 \item On the other hand, we can find a weak Condorcet winner using only $\BigO(n\log m)$ many queries for single peaked profiles [\Cref{thm:con_sp_ub}]. Moreover, we show this bound is tight under the condition that the algorithm does not interleave queries to different voters [\Cref{thm:con_sp_lb}] (our algorithm indeed satisfies this condition). For any arbitrary underlying single peaked tree \TT, we provide an algorithm for finding a weak Condorcet winner with query complexity $\BigO(nk\log m)$ where $k$ is the path cover number of \TT [\Cref{thm:cw_pc_ub}] or the number of leaves in \TT [\Cref{cor:cw_leaf_ub}].
\end{itemize}}

To summarize, we remark that our results non-trivially generalize earlier works on query complexity for preference elicitation in \cite{Conitzer09}. We believe revisiting the preference elicitation problem in the context of profiles that are single peaked on trees is timely, and that this work also provides fresh algorithmic and structural insights on the domain of preferences that are single peaked on trees.

\longversion{\subsection{Related Work}}
\shortversion{\paragraph*{Related Work.}} We have already mentioned the work of~\cite{Conitzer09} addressing the question of eliciting preferences in single-peaked profiles, which is the closest predecessor to our work. Before this, Conitzer and Sandholm addressed the computational hardness for querying minimally for winner determination \cite{conitzer2002vote}. They also prove that one would need to make $\Omega(mn\log m)$ queries even to decide the winner for many commonly used voting rules \cite{conitzer2005communication} which matches with the trivial $\BigO(mn\log m)$ upper bound for preference elicitation in unrestricted domain based on sorting. Ding and Lin study preference elicitation under partial information setting and show interesting properties of what they call a deciding set of queries \cite{ding2013voting}. Lu and Boutilier provide empirical study of preference elicitation under probabilistic preference model \cite{lu2011vote} and devise several novel heuristics which often work well in practice \cite{LuB11a}.


%
%
%
%
%
%
%
\longversion{
The rest of the paper is organized as follows: We introduce the basic terminologies and formally define the problems in \Cref{sec:prelim}, present results for eliciting preference profile and finding a weak Condorcet winner in \Cref{sec:pe} and \Cref{sec:cw} respectively, and conclude with future works in \Cref{sec:com}.
}
\section{Preliminaries}\label{sec:prelim}

For a positive integer $n$, we denote the set $\{1, \ldots, n\}$ by $[n]$. Suppose we have a set $\VV=\{v_i: i\in[n]\}$ of $n$ {\em voters} each of which has a preference $\succ_i, i\in[n]$, alternatively called {\em vote}, which is a complete order over a set $\CC=\{c_j: j\in[m]\}$ of $m$ {\em candidates}. We denote the set of preferences over \CC by $\LL(\CC)$. The tuple $(\succ_i)_{i\in[n]}$ of the preferences of all the voters is called a {\em profile}. If not mentioned otherwise, we use $m$ and $n$ to denote the number of candidates and the number of voters respectively. Let $\succ\thinspace\in\LL(\CC)$ be any complete order over \CC. For a subset $\XX\subseteq\CC$ of candidates, we denote the restriction of an order $\succ$ to the subset of candidates $\XX$ by $\succ(\XX)$. We denote the restriction of a profile $\PP = (\succ_i)_{i\in[n]}$ to \XX by $\PP(\XX) = (\succ_i(\XX))_{i\in[n]}$. We say that a candidate $x\in\CC$ is placed at the $i^{th}$ position of a preference $\succ\thinspace\in\LL(\CC)$ if $x$ is preferred over all but exactly $(i-1)$ candidates in $\succ$. Given an $n$ voters profile $\PP = (\succ_i)_{i\in[n]}$, a candidate $x$ is called a {\em Condorcet winner} of \PP if for every other candidate $y$, a {\em strict} majority of the voters prefer $x$ over $y$; that is $|\{i\in[n]: x\succ_i y\}| > \nfrac{n}{2}$ for every $y\in\CC\setminus\{x\}$. A candidate $x$ is called a {\em weak Condorcet winner} of an $n$ voters profile $\PP = (\succ_i)_{i\in[n]}$ if there does not exist any other candidate $y$ whom a strict majority of the voters prefer over $x$; that is $|\{i\in[n]: y\succ_i x\}| \le \nfrac{n}{2}$ for every $y\in\CC\setminus\{x\}$.
\longversion{For an order $\succ\thinspace\in\LL(\CC)$, we denote the ``reverse order'' of $\succ$ by $\overleftarrow{\succ}$; that is $\overleftarrow{\succ} = \{x\overleftarrow{\succ}y : x, y\in\CC \text{ and } y\succ x\}$.}
A preference $\succ\thinspace\in\LL(\CC)$ over a set of candidates \CC is called {\em single peaked} with respect to an order $\succ^\pr\in\LL(\CC)$ if, for every candidates $x, y\in\CC$, we have $x\succ y$ whenever we have either $c\succ^\pr x\succ^\pr y$ or $y\succ^\pr x\succ^\pr c$, where $c\in\CC$ is the candidate at the first position of $\succ$. A profile $\PP=(\succ_i)_{i\in[n]}$ is called single peaked with respect to an order $\succ^\pr\in\LL(\CC)$ if $\succ_i$ is single peaked with respect to $\succ^\pr$ for every $i\in[n]$. \longversion{ Notice that if a profile \PP is single peaked with respect to an order $\succ^\pr\in\LL(\CC)$, then \PP is also single peaked with respect to the order $\overleftarrow{\succ}^\pr$.} Given a path $\QQ = (x_1, x_2, \ldots, x_\el)$ from a vertex $x_1$ to another vertex $x_\el$ in a tree \TT, we define the order induced by the path \QQ to be $x_1\succ x_2\succ \cdots\succ x_\el$. Given a tree $\TT = (\CC, \EE)$ with the set of nodes as the set of candidates \CC, a profile \PP is called single peaked on the tree \TT if \PP is single peaked on the order induced by every path of the tree \TT; that is for every two candidates $x, y\in\CC$, the profile $\PP(\XX)$ is single peaked with respect to the order $\succ$ of the candidates \XX induced by the unique path from $x$ to $y$ in \TT. We call the tree \TT the {\em underlying single peaked tree}. It is known (c.f.~\cite{demange1982single}) that there always exists a weakly Condorcet winner for a profile \PP which is single peaked on a tree \TT.

\paragraph{Trees.} The following definitions pertaining to the structural aspects of trees will be useful.
\begin{itemize}
	\item The {\em pathwidth} of \TT is the minimum width of a {\em path decomposition} of \TT \cite{heinrich1993path}. 
	\item A set of disjoint paths $\QQ = \{Q_1 = (\XX_1, \EE_1), \ldots, Q_k = (\XX_k, \EE_k)\}$ is said to cover a tree $\TT = (\XX, \EE)$ if $\XX = \cup_{i\in[k]} \XX_i, \EE_i\subseteq \EE, \XX_i \cap \XX_j = \emptyset, \EE_i \cap \EE_j = \emptyset$ for every $i, j\in[k]$ with $i\ne j$. The {\em path cover number} of \TT is the cardinality of the smallest set \QQ of {\em disjoint} paths that cover \TT.
	\item The {\em distance of a tree \TT from a path} is the smallest number of nodes whose removal makes the tree a path.
	\item The {\em diameter} of a tree \TT is the number of edges in the longest path in \TT.
\end{itemize}

We also list some definitions of \textit{subclasses} of trees (which are special types of trees, see also~\Cref{fig:tree-classes}). 
\begin{itemize}
\item A tree is a {\em star} if there is a {\em center vertex} and every other vertex is a neighbor of this vertex.
\item A tree is a {\em subdivision of a star} if it can be constructed by replacing each edge of a star by a path. 
\item A subdivision of a star is called {\em balanced} if there exists an integer \el such that the distance of every leaf node from the center is \el. 
\item A tree is a {\em caterpillar} if there is a {\em central path} and every other vertex is at a distance of one it. 
\item A tree is a {\em complete binary tree} rooted at $r$ if every nonleaf node has exactly two children and there exists an integer $h$, called the height of the tree, such that every leaf node is at a distance of either $h$ or $h-1$ from the root node $r$.	
\end{itemize}

\begin{figure}[t]
	\centering
	\includegraphics[scale=0.35]{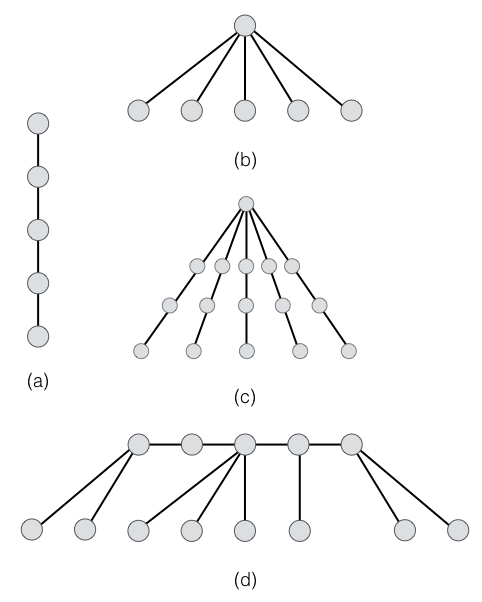}
	\caption{Depicting classes of trees: (a) a path, (b) a star, (c) a (balanced) subdivision of a star, (d) a caterpillar.}
	\label{fig:tree-classes}
\end{figure}



\paragraph*{Problem Definitions and Known Results.}~Suppose we have a profile \PP with $n$ voters and $m$ candidates. For any pair of distinct candidates $x$ and $y$, and a voter $\el \in [n]$, we introduce the boolean-valued function $\text{\Query}(x \succ_\el y)$. The output of this function is \true if the voter \el prefers the candidate $x$ over the candidate $y$ and \false otherwise. We now formally state the two problems that we consider in this paper.

\begin{definition}\PE\\
 Given a tree $\TT = (\CC, \EE)$ and an oracle access to the function \Query($\cdot$) for a profile \PP which is single peaked on \TT, find \PP.
\end{definition}

\begin{definition}\CW\\
 Given a tree $\TT = (\CC, \EE)$ and an oracle access to the function \Query($\cdot$) for a profile \PP which is single peaked on \TT, find a weak Condorcet winner of \PP.
\end{definition}

Suppose we have a set of candidates $\CC = \{c_1, \ldots, c_m\}$. We say that an algorithm \AA makes $q$ queries if there are exactly $q$ distinct tuples $(\el, c_i, c_j)\in[n]\times\CC\times\CC$ with $i<j$ such that \AA calls \Query($c_i\succ_\el c_j$) or \Query($c_j\succ_\el c_i$). We call the number of queries made by an algorithm \AA its {\em query complexity}. 

We state some known results that we will appeal to later. The first observation employs a sorting algorithm like merge sort to elicit every vote with $\BigO(m\log m)$ queries, while the second follows from the linear-time merge subroutine of merge sort (\cite{cormen2009introduction}).


\begin{observation}\label{obs:trivial}
 There is a \PE algorithm with query complexity $\BigO(mn\log m)$.
\end{observation}

\begin{observation}\label{obs:merge}
 Suppose $\CC_1, \CC_2\subseteq\CC$ form a partition of \CC and $\succ$ is a ranking of the candidates in \CC. Then there is a polynomial time algorithm that finds $\succ$ given $\succ(\CC_1)$ and $\succ(\CC_2)$ with query complexity $\BigO(|\CC|)$.
\end{observation}


\begin{theorem}\label{thm:con}\cite{Conitzer09} There is a \PE algorithm with query complexity $\BigO(mn)$ for single peaked profiles.
\end{theorem}

\longversion{We now state the \CW problem, which asks for eliciting only up to the point of determining a weak Condercet winner (recall that at least one such winner is guaranteed to exist on profiles that are single-peaked on trees). 
\begin{definition}\CW\\
 Given a tree $\TT = (\CC, \EE)$ and an oracle access to the function \Query($\cdot$) for a profile \PP which is single peaked on \TT, find a weak Condorcet winner of \PP.
\end{definition}}


\section{Results for \PE}\label{sec:pe}
In this section, we present our results for \PE for profiles that are single peaked on trees. Recall that we would like to generalize~\Cref{thm:con} in a way to profiles that are single peaked on trees. Since the usual single peaked profiles can be viewed as profiles single peaked with respect to a path, we propose the following measures of how much a tree resembles a path.

 \begin{itemize}
 	\item \textit{Leaves.} Recall any tree has at least two leaves, and paths are the trees that have exactly two leaves. We consider the class of trees that have $\el$ leaves, and show an algorithm with a query complexity of $\BigO(mn\log \el)$. 
 	\item \textit{Path Cover.} Consider the notion of a \textit{path cover number} of a tree, which is the smallest number of disjoint paths that the tree can be partitioned into. Clearly, the path cover number of a path is one; and for trees that can be covered with $k$ paths, we show an algorithm with query complexity $\BigO(mn\log k)$.
 	\item \textit{Distance from Paths.} Let $d$ be the size of the smallest set of vertices whose removal makes the tree a path. Again, if the tree is a path, then the said set is simply the empty set. For trees that are at a distance $d$ from being a path (in the sense of vertex deletion), we provide an algorithm with query complexity $\BigO(mn\log d)$.
 	\item \textit{Pathwidth and Maximum Degree.} Finally, we note that paths are also trees that have pathwidth one, and maximum degree two. These perspectives turn out to be less useful: in particular, there are trees where these parameters are constant, for which we show that elicitation is as hard as it would be on an arbitrary profile, and therefore the easy algorithm from~\Cref{obs:trivial} is actually the best that we can hope for. 
 \end{itemize}

For the first three perspectives that we employ, that seemingly capture an appropriate aspect of paths and carry it forward to trees, the query complexities that we obtain are tight --- we have matching lower bounds in all cases. Also, while considering structural parameters, it is natural to wonder if there is a class of trees that are incomparable with paths but effective for elicitation. Our attempt in this direction is to consider trees of bounded diameter. However, again, we find that this is not useful, as we have examples to show that there exist trees of diameter two that are as hard to elicit as general profiles. 

We remark at this point that all these parameters are polynomially computable for trees, making the algorithmic results viable. Also, for the parameters of pathwidth, maximum degree and diameter, we show lower bounds on trees where these parameters are large (such as trees with pathwidth $O(\log m)$, maximum degree $m-1$, and diameter $m/2$), which --- roughly speaking --- also rules out the possibility of getting a good inverse dependence. As a concrete example, motivated by the $\BigO(mn)$ algorithm for paths, which have diameter $m$, one might wonder if there is an algorithm with query complexity $\BigO(\frac{mn \log m}{\log \omega})$. This possibility, in particular, is ruled out. 
We are now ready to discuss the results in~\Cref{tbl:elicit}. 

\longversion{
We next show a structural result about trees namely any tree with \el leaves can be partitioned into \el paths. The idea is to fix some nonleaf node as root and iteratively find a path from low depth nodes (depth of a node is its distance from root) to some leaf node which is disjoint from all the paths chosen so far. We formalize this idea below.

\begin{lemma}\label{lem:path_decom}
 Let $\TT = (\XX, \EE)$ be a tree with \el leaves. Then there is a polynomial time algorithm which partitions \TT into \el disjoint paths $\QQ_i = (\XX_i, \EE_i), i\in[\el]$; that is we have $\XX_i \cap \XX_j = \emptyset, \EE_i \cap \EE_j = \emptyset$ for every $i, j\in[\el]$ with $i\ne j$, $\XX = \cup_{i\in[\el]} \XX_i$, $\EE = \cup_{i\in[\el]} \EE_i$, and $\QQ_i$ is a path in \TT.
\end{lemma}

\longversion{
\begin{proof}
 We first make the tree \TT rooted at any arbitrary nonleaf node $r$. We now partition the tree \TT into paths iteratively as follows. Initially every node of the tree \TT is {\em unmarked} and the set of paths \QQ we have is empty. Let $Q_1 = (\XX_1, \EE_1)$ be the path in \TT from the root node to any leaf node. We put $Q_1$ into \QQ and {\em mark} all the nodes in $Q_1$. More generally, in the $i^{th}$ iteration we pick an unmarked node $u$ that is closest to the root $r$, breaking ties arbitrarily, and add any path $\QQ_i$ in \TT from $u$ to any leaf node in the subtree $\TT_u$ rooted at $u$, and {\em mark} all the nodes in $\QQ_i = (\XX_i, \EE_i)$. Since $u$ is unmarked, we claim that every node in $\TT_u$ is unmarked. Indeed, otherwise suppose there is a node $w$ in $\TT_u$ which is already marked. Then there exists two paths from $r$ to $w$ one including $u$ and another avoiding $u$ since $u$ is currently unmarked and $w$ is already marked. This contradicts the fact that \TT is a tree. Hence every node in $\TT_u$ is unmarked. We continue until all the leaf nodes are marked and return \QQ. Since \TT has \el leaves and in every iteration at least one leaf node is marked (since every $\QQ_i$ contains a leaf node), the algorithm runs for at most \el iterations. Notice that, since the algorithm always picks a path consisting of unmarked vertices only, the set of paths in \QQ are pairwise disjoint. We claim that \QQ forms a partition of \TT. Indeed, otherwise there must be a node $x$ in \TT that remains unmarked at the end. From the claim above, we have all the nodes in the subtree $\TT_x$ rooted at $x$ unmarked which includes at least one leaf node of \TT. This contradicts the fact that the algorithm terminates when all the leaf nodes are marked.
\end{proof}
}

Using \Cref{lem:path_decom}, and the fact that any path can account for at most two leaves, we have that the path cover number of a tree is same as the number of leaves up to a factor of two.

\begin{lemma}\label{lem:path_leaf}
 Suppose the path cover number of a tree \TT with \el leaves is $k$. Then we have $\nfrac{\el}{2}\le k\le \el$.
\end{lemma}

\longversion{
\begin{proof}
 The inequality $\nfrac{\el}{2}\le k$ follows from the fact that any path in \TT can involve at most two leaves in \TT and there exists $k$ paths covering all the leaf nodes. The inequality $k\le \el$ follows from the fact from \Cref{lem:path_decom} that any tree \TT with \el leaves can be partitioned into \el paths.
\end{proof}
}

}

\subsection{Algorithmic Results}

We now present our main algorithmic results. We begin with generalizing the result of \Cref{thm:con} to any single peaked profiles on trees whose path cover number is at most $k$. The idea is to partition the tree into $k$ disjoint paths, use the algorithm from \Cref{thm:con} on each paths to obtain an order of the candidates on each path of the partition, and finally merge these suborders intelligently. We now formalize this idea as follows.

\begin{theorem}\label{thm:pcover_elicit_ub}
 There is a \PE algorithm with query complexity $\BigO(mn\log k)$ for profiles that are single peaked on trees with path cover number at most $k$.
\end{theorem}

\begin{proof}
 Since the path cover number is at most $k$, We can partition the tree $\TT = (\CC, \EE)$ into $t$ disjoint paths $\PP_i = (\CC_i, \EE_i), i\in[t]$, where $t$ is at most $k$. We now show that we can elicit any preference $\succ$ which is single peaked on the tree \TT by making $\BigO(m\log t)$ queries which in turn proves the statement. We first find the preference ordering restricted to $\CC_i$ using \Cref{thm:con} by making $\BigO(|\CC_i|)$ queries for every $i\in[t]$. This step needs $\sum_{i\in[t]} \BigO(|\CC_i|) = \BigO(m)$ queries since $\CC_i, i\in[t]$ forms a {\em partition} of \CC. We next merge the $t$ orders $\succ(\CC_i), i\in[t],$ to obtain the complete preference $\succ$ by using a standard divide and conquer approach for $t$-way merging which makes $O(m \log t)$ queries \cite{hopcroft1983data}. Thus the query complexity of the algorithms is $\BigO(m + m\log t) = \BigO(m\log k)$.
\end{proof}

Towards our results on leaves, we show that the path cover of a tree with $\el$ leaves is at least $\lfloor\el/2 \rfloor$ and at most $\el$. The lower bound is easy since every path can account for at most two leaves. The upper bound comes from a careful partitioning of the nodes following maximal paths from leaves up to a root and using a careful marking scheme. We defer the details of this argument to a full version. Using the fact that a tree with $\el$ leaves can be partitioned into at most $\el$ paths, we can use the algorithm from~\Cref{thm:pcover_elicit_ub} to obtain the following bound in terms of leaves.

 \longversion{Initially we have $t$ orders to merge. We arbitrarily pair the $t$ orders into $\lceil \nfrac{t}{2}\rceil$ pairs with at most one of them being singleton (when $t$ is odd). By renaming, suppose the pairings are as follows: $(\succ(\CC_{2i-1}), \succ(\CC_{2i})), i\in[\lfloor\nfrac{t}{2}\rfloor]$. Let us define $\CC_i^\pr = \CC_{2i-1} \cup \CC_{2i}$ for every $i\in[\lfloor\nfrac{t}{2}\rfloor]$ and $\CC_{\lceil\nfrac{t}{2}\rceil}^\pr = \CC_t$ if $t$ is an odd integer. We merge $\succ(\CC_{2i-1})$ and $\succ(\CC_{2i})$ to get $\succ(\CC_i^\pr)$ for every $i\in[\lfloor\nfrac{t}{2}\rfloor]$ using \Cref{obs:merge}. The number queries the algorithm makes in this iteration is $\sum_{i\in[\lfloor\nfrac{t}{2}\rfloor]} \BigO(|\CC_{2i-1}| + |\CC_{2i}|) = \BigO(m)$ since $(\CC_i)_{i\in[t]}$ forms a partition of \CC. At the end of the first iteration, we have $\lceil\nfrac{t}{2}\rceil$ orders $\succ(\CC_i^\pr), i\in[\lceil\nfrac{t}{2}\rceil]$ to merge to get $\succ$. The algorithm repeats the step above $\BigO(\log t)$ times to obtain $\succ$ and the query complexity of each iteration is $\BigO(m)$. } 
 
\begin{corollary}\label{cor:leaves_elicit_ub}
 There is a \PE algorithm with query complexity $\BigO(mn\log \el)$ for profiles that are single peaked on trees with at most $\el$ leaves.
\end{corollary}

Finally, if we are given a subset of vertices whose removal makes the given tree a path, then we have an elicitation algorithm that makes $\BigO(mn + nd\log d)$. As before, we can determine the ordering among the candidates on the path with $\BigO(m-d)$ queries, and we determine the ordering among the rest in $\BigO(d\log d)$ queries using~\Cref{obs:trivial}, and finally merge using~\Cref{obs:merge}. This leads us to the following. \shortversion{(We defer the proofs of the results which are marked star to full version of the paper.)}

\begin{theorem}\label{thm:pdist_elicit_ub}
 There is a \PE algorithm with query complexity $\BigO(mn + nd\log d)$ for profiles that are single peaked on trees with distance $d$ from path.
\end{theorem}

\longversion{
\begin{proof}
 Let \XX be the smallest set of nodes of a tree $\TT = (\CC, \EE)$ such that $\TT\setminus\XX$, the subgraph of \TT after removal of the nodes in \XX, is a path. We have $|\XX|\le d$. For any preference $\succ$, we make $\BigO(d\log d)$ queries to find $\succ(\XX)$ using \Cref{obs:trivial}, make $\BigO(|\CC\setminus\XX|) = \BigO(m-d)$ queries to find $\succ(\CC\setminus\XX)$ using \Cref{thm:con}, and finally make $\BigO(m)$ queries to find $\succ$ by merging $\succ(\XX)$ and $\succ(\CC\setminus\XX)$ using \Cref{obs:merge}. This gives an overall query complexity of $\BigO(mn + nd\log d)$.
\end{proof}
}

\subsection{Lower Bounds}

We now turn to the lower bounds. Our first result is based on a counting argument, showing that the query complexity in terms of the number of leaves, given by \Cref{cor:leaves_elicit_ub}, is tight up to constant factors. Indeed, let us consider a subdivision of a star with $\el$ leaves and let $t$ denote the distance from the center, so that we have a total of $t\el + 1$ vertices. One can show that if the candidates are written out in level order, the candidates that are distance $i$ from the star can be ordered arbitrarily within this ordering. This tells us that the number of possible preferences $\succ$ that are single peaked on the tree \TT is at least $(\el!)^t$. We obtain the lower bound by using a decision tree argument, wherein we are able to show that it is always possible for an oracle answering the comparison queries to answer in such a way that the total space of possibilities decreases by at most a factor of half. Since the tree must entertain at least $(\el!)^t$ leaves to account for all possibilities, we obtain the claimed lower bound.

\begin{theorem}\label{thm:leaves_elicit_lb}\shortversion{$[\star]$}
 Let $\TT = (\CC, \EE)$ be a balanced subdivision of a star with $\el$ leaves. Then any \PE algorithm for single peaked profiles on \TT has query complexity $\Omega(mn\log \el)$.
\end{theorem}
\longversion{
\begin{proof}
 Suppose the number of candidates $m$ be $(t\el + 1)$ for some integer $t$. Let $c$ be the center of \TT. We denote the shortest path distance between any two nodes $x, y\in\CC$ in \TT by $d(x, y)$. We consider the partition $(\CC_0, \ldots, \CC_t)$ of the set of candidates \CC where $\CC_i = \{ x\in\CC : d(x, c) = i\}$. We claim that the preference $\succ = \pi_0 \succ \pi_1 \succ \cdots \succ \pi_t$ of the set of candidates \CC is single peaked on the tree \TT where $\pi_i$ is any arbitrary order of the candidates in $\CC_i$ for every $0\le i\le t$. Indeed; consider any path $\QQ = (\XX, \EE^\pr)$ in the tree \TT. Let $y$ be the candidate closest to $c$ among the candidates in \XX; that is $y = \argmin_{x\in\XX} d(x,c)$. Then clearly $\succ(\XX)$ is single peaked with respect to the path \QQ having peak at $y$. We have $|\CC_i| = \el$ for every $i\in[t]$ and thus the number of possible preferences $\succ$ that are single peaked on the tree \TT is at least $(\el!)^t$. 
 
 Let \AA be any \PE algorithm for single peaked profiles on the tree \TT. We now describe our oracle to answer the queries that the algorithm \AA makes. For every voter $v$, the oracle maintains the set $\RR_v$ of possible preferences of the voter $v$ which is consistent with the answers to all the queries that the algorithm \AA have already made for the voter $v$. At the beginning, we have $|\RR_v| \ge ( \el!)^t$ for every voter $v$ as argued above. Whenever the oracle receives a query on $v$ for any two candidates $x$ and $y$, it computes the numbers $n_1$ and $n_2$ of orders in $\RR_v$ which prefers $x$ over $y$ and $y$ over $x$ respectively; the oracle can compute the integers $n_1$ and $n_2$ since the oracle has infinite computational power. The oracle answers that the voter $v$ prefers the candidate $x$ over $y$ if and only if $n_1\ge n_2$ and updates the set $\RR_v$ accordingly. Hence, whenever the oracle is queried for a voter $v$, the size of the set $\RR_v$ decreases by a factor of at most two. On the other hand, we must have, from the correctness of the algorithm, $\RR_v$ to be a singleton set when the algorithm terminates for every voter $v$ -- otherwise there exists a voter $v$ (whose corresponding $\RR_v$ is not singleton) for which there exist two possible preferences which are single peaked on the tree \TT and are consistent with all the answers the oracle has given and thus the algorithm \AA fails to output the preference of the voter $v$ correctly. Hence every voter must be queried at least $\Omega(\log((\el!)^t)) = \Omega(t \el\log \el) = \Omega(m\log \el)$ times.
\end{proof}
}

Since the path cover number of a subdivided star on $\el$ leaves is at least 
$\el/2$, we also have the following.

\begin{corollary}\label{cor:pcover_elicit_lb}
 There exists a tree \TT with path cover number $k$ such that any \PE algorithm for single peaked profiles on \TT has query complexity $\Omega(mn\log k)$.
\end{corollary}

Mimicking the level order argument above on a generic tree with $\el$ leaves, and using the connection between path cover and leaves, we obtain lower bounds that are functions of $(n,\el)$ and $(n,k)$, as given below. This will be useful for our subsequent results. 


\begin{theorem}\label{thm:pcover_elicit_lb_any}\shortversion{$[\star]$}
 Let $\TT = (\CC, \EE)$ be any arbitrary tree with $\el$ leaves and path cover number $k$. Then any \PE algorithm for single peaked profiles on \TT has query complexity $\Omega(n\el\log \el)$ and $\Omega(nk\log k)$.
\end{theorem}

\longversion{
\begin{proof}
 Let \XX be the set of leaves in \TT. We choose any arbitrary nonleaf node $r$ as the root of \TT. We denote the shortest path distance between two candidates $x, y\in\CC$ in the tree \TT by $d(x,y)$. Let $t$ be the maximum distance of a node from $r$ in \TT; that is $t = \max_{y\in\CC\setminus\XX} d(r, x)$. We partition the candidates in $\CC\setminus\XX$ as $(\CC_0, \CC_1, \ldots, \CC_t)$ where $\CC_i = \{y\in\CC\setminus\XX : d(r, y) = i\}$ for $0\le i\le t$. We claim that the preference $\succ = \pi_0 \succ \pi_1 \succ \cdots \succ \pi_t \succ \pi$ of the set of candidates \CC is single peaked on the tree \TT where $\pi_i$ is any arbitrary order of the candidates in $\CC_i$ for every $0\le i\le t$ and $\pi$ is an arbitrary order of the candidates in $\CC\setminus\XX$. Indeed, otherwise consider any path $\QQ = (\YY, \EE^\pr)$ in the tree \TT. Let $y$ be the candidate closest to $r$ among the candidates in \YY; that is $y = \argmin_{x\in\YY} d(x,r)$. Then clearly $\succ(\YY)$ is single peaked with respect to the path \QQ having peak at $y$. We have the number of possible preferences $\succ$ that are single peaked on the tree \TT is at least $|\XX|! = \el!$. Again using the oracle same as in the proof of \Cref{thm:leaves_elicit_lb}, we deduce that any \PE algorithm \AA for profiles that are single peaked on the tree \TT needs to make $\Omega(n\el\log \el)$ queries. The bound with respect to the path cover number now follows from \Cref{lem:path_leaf}.
\end{proof}
}

The following results can be obtained simply by applying~\Cref{thm:pcover_elicit_lb_any} on particular graphs. For instance, we use the fact that stars have $(m-1)$ leaves and have pathwidth one to obtain the first part of \Cref{cor:pwidth_elicit_lb}, while appealing to complete binary trees that have $O(m)$ leaves and pathwidth $O(\log m)$ for the second part. These examples also work in the context of maximum degree, while for diameter we use stars and caterpillars with a central path of length $m/2$.


\longversion{We next consider the parameter pathwidth of the underlying single peaked tree. We immediately get the following result for \PE on trees with pathwidths one or $\log m$ from \Cref{thm:pcover_elicit_lb_any} and the fact that the pathwidths of a star and a complete binary tree are one and $\log m$ respectively.}

\begin{corollary}\label{cor:pwidth_elicit_lb}
 There exist two trees \TT and $\TT^\pr$ with pathwidths one and $\log m$ respectively such that any \PE algorithm for single peaked profiles on \TT and $\TT^\pr$ respectively has query complexity $\Omega(mn\log m)$.
\end{corollary}


\begin{corollary}\label{cor:maxdeg_elicit_lb}
 There exist two trees \TT and $\TT^\pr$ with maximum degree $\Delta=3$ and $m-1$ respectively such that any \PE algorithm for single peaked profiles on \TT and $\TT^\pr$ respectively has query complexity $\Omega(mn\log m)$.
\end{corollary}

\longversion{
\begin{proof}
 Using \Cref{thm:pcover_elicit_lb_any}, we know that any \PE algorithm for profiles which are single peaked on a complete binary tree has query complexity $\Omega(mn\log m)$ since a complete binary tree has $\Omega(m)$ leaves. The result now follows from the fact that the maximum degree $\Delta$ of a node is three for any binary tree. The case of $\Delta = m-1$ follows immediately from \Cref{thm:pcover_elicit_lb_any} applied on stars.
\end{proof}
}


\begin{corollary}\label{cor:dia_elicit_lb}
 There exists two trees \TT and $\TT^\pr$ with diameters $\omega=2$ and $\omega=\nfrac{m}{2}$ respectively such that any \PE algorithm for profiles which are single peaked on \TT and $\TT^\pr$ respectively has query complexity $\Omega(mn\log m)$.
\end{corollary}

\longversion{
\begin{proof}
 The $\omega=2$ and $\omega=\nfrac{m}{2}$ cases follow from \Cref{thm:pcover_elicit_lb_any} applied on star and caterpillar graphs with a central path of length $\nfrac{m}{2}$ respectively.
\end{proof}
}


Our final result, which again follows from \Cref{thm:pcover_elicit_lb_any} applied of caterpillar graphs with a central path of length $m-d$, shows that the bound in~\Cref{thm:pdist_elicit_ub} is tight. 

\begin{theorem}\label{thm:pdist_elicit_lb}\shortversion{$[\star]$}
 For any integers $m$ and $d$ with $1\le d\le\nfrac{m}{4}$, there exists a tree \TT with distance $d$ from path such that any \PE algorithm for profiles which are single peaked on \TT has query complexity $\Omega(mn + nd\log d)$.
\end{theorem}

\longversion{
\begin{proof}
 Consider the caterpillar graph where the length of the central path \QQ is $m-d$; there exists such a caterpillar graph since $d\le\nfrac{m}{4}$. Consider the order $\succ = \pi \succ \sigma$ of the set of candidates \CC where $\pi$ is an order of the candidates in \QQ which is single peaked on \QQ and $\sigma$ is any order of the candidates in $\CC\setminus\QQ$. Clearly, $\succ$ is single peaked on the tree \TT. Any elicitation algorithm \AA needs to make $\Omega(m-d)$ queries involving only the candidates in \QQ to elicit $\pi$ due to \cite{Conitzer09} and $\Omega(d\log d)$ queries to elicit $\sigma$ due to sorting lower bound for every voter. This proves the statement.
\end{proof}
}

\section{\CW}\label{sec:cw}

We now show that we can find a weak Condorcet winner of profiles that are single peaked on trees using fewer queries than the number of queries needed to find the profile itself. First, note that if a Condorcet winner is guaranteed to exist, then it can be found using $\BigO(mn)$ queries --- we pit an arbitrary pair of candidates $x,y$ and use $\BigO(n)$ queries to determine if $x$ defeats $y$. We push the winning candidate forward and repeat the procedure, clearly requiring at most $m$ rounds. Now, if a profile is single peaked with respect to a tree, and there are an odd number of voters, then we have a Condorcet winner and the procedure that we just described would work. Otherwise, we simply find a Condorcet winner among the first $(n-1)$ voters. It can be shown that such a winner is one of the weak Condorcet winners for the overall profile, and we therefore have the following upper bound. 

\longversion{
We begin with the following general observation.

\begin{observation}\label{obs:cond_trivial_ub}
 Let \PP be a profile where a Condorcet winner is guaranteed to exist. Then we can find the Condorcet winner of \PP by making $\BigO(mn)$ queries.
\end{observation}

\begin{proof}
 For any two candidates $x, y\in\CC$ we find whether $x$ defeats $y$ or not by simply asking all the voters to compare $x$ and $y$; this takes $\BigO(n)$ queries. The algorithms maintains a set \SS of candidates which are {\em potential} Condorcet winners. We initialize \SS to \CC. In each iteration we pick any two candidates $x, y\in\SS$ from \SS, remove $x$ from \SS if $x$ does not defeat $y$ and vice versa using $\BigO(n)$ query complexity until \SS is singleton. After at most $m-1$ iterations, the set \SS will be singleton and contain only Condorcet winner since we find a candidate which is not a Condorcet winner in every iteration and thus the size of the set \SS decreases by at least one in every iteration. This gives a query complexity bound of $\BigO(mn)$.
\end{proof}

Using \Cref{obs:cond_trivial_ub} we now develop a \CW algorithm with query complexity $\BigO(mn)$ for profiles that are single peaked on trees.
}

\begin{theorem}\label{thm:cw_gen_ub}\shortversion{$[\star]$}
 There is a \CW algorithm with query complexity $\BigO(mn)$ for single peaked profiles on trees.
\end{theorem}

For the special case of single peaked profiles, we can do even better. Here we take advantage of the fact that a ``median candidate''~\cite{mas1995microeconomic}  is guaranteed to be a weak Condorcet winner. We make $\OO(\log m)$ queries per vote to find the candidates placed at the first position of all the votes using the algorithm in \cite{Conitzer09} and find a median candidate to show the following.

\longversion{ if a profile is single peaked (on a path), then there is a \CW algorithm with query complexity $\BigO(n\log m)$ as shown below. Let us define the frequency $f(x)$ of a candidate $x\in\CC$ to be the number of votes where $x$ is placed at the first position. Then we know that a median candidate according to the single peaked ordering of the candidates along with their frequencies as defined above is a weak Condorcet winner for single peaked profiles \cite{mas1995microeconomic}.}

\begin{theorem}\label{thm:con_sp_ub}\shortversion{$[\star]$}
 There is a \CW algorithm with query complexity $\BigO(n\log m)$ for single peaked profiles (on a path).
\end{theorem}

\longversion{
\begin{proof}
 Let \PP be a profile that is single peaked with respect to an ordering $\succ\in\LL(\CC)$ of candidates. Then we find, for every voter $v$, the candidate the voter $v$ places at the first position using $\BigO(\log m)$ queries using the algorithm in \cite{Conitzer09} and return a median candidate.
\end{proof}
}

\longversion{
\begin{proof}
 Let $\PP = (\succ_i)_{i\in[n]}$ be a profile which is single peaked on a tree \TT. If $n$ is an odd integer, then we know that there exists a Condorcet winner in \PP since no two candidates can tie and there always exists at least one weak Condorcet winner in every single peaked profile on trees. Hence, if $n$ is an odd integer, then we use \Cref{obs:cond_trivial_ub} to find a weak Condorcet winner which is the Condorcet winner too. Hence let us now assume that $n$ is an even integer. Notice that $\PP_{-1} = (\succ_2, \ldots, \succ_n)$ is also single peaked on \TT and has an odd number of voters and thus has a Condorcet winner. We use \Cref{obs:cond_trivial_ub} to find the Condorcet winner $c$ of $\PP_{-1}$ and output $c$ as a weak Condorcet winner of \PP. We claim that $c$ is a weak Condorcet winner of \PP. Indeed otherwise there exists a candidate $x$ other than $c$ who defeats $c$ in \PP. Since $n$ is an even integer, $x$ must defeat $c$ by a margin of at least two (since all pairwise margins are even integers) in \PP. But then $x$ also defeats $c$ by a margin of at least one in $\PP_{-1}$. This contradicts the fact that $c$ is the Condorcet winner of $\PP_{-1}$.
\end{proof}
}
The next result uses~\Cref{thm:con_sp_ub} on paths in a path cover eliminating the case of even number of voters by the idea of setting aside one voter that was used in~\Cref{thm:cw_gen_ub}.

\begin{theorem}\label{thm:cw_pc_ub}\shortversion{$[\star]$}
 Let \TT be a tree with path cover number at most $k$. Then there is an algorithm for \CW for profiles which are single peaked on \TT with query complexity $\BigO(nk\log m)$.
\end{theorem}

Recalling that the number of leaves bounds the path cover number, we have the following consequence. 

\begin{corollary}\label{cor:cw_leaf_ub}
 Let \TT be a tree with \el leaves. Then there is an algorithm for \CW for profiles which are single peaked on \TT with query complexity $\BigO(n\el\log m)$.
\end{corollary}

We now state the lower bounds pertaining to \CW. First, we show that any algorithm for single peaked profiles on stars has query complexity $\Omega(mn)$, showing that the bound of \Cref{thm:cw_gen_ub} is tight. 

\begin{theorem}\label{thm:cw_gen_lb}
 Any \CW algorithm for single peaked profiles on stars must have query complexity $\Omega(mn)$.
\end{theorem}

\begin{proof}
 Let \TT be a star with center vertex $c$. We now design an oracle that will ``force'' any \CW algorithm \AA for single peaked profiles on \TT to make $\Omega(mn)$ queries. For every voter $v$, the oracle maintains a set of {\em ``marked''} candidates which can not be placed at the first position of the preference of $v$. Suppose the oracle receives a query to compare two candidates $x$ and $y$ for a voter \el. If the order between $x$ and $y$ for the voter \el follows from the answers the oracle has already provided to all the queries for the voter \el, then the oracle answers accordingly. Otherwise it answers $x\succ_\el y$ if $y$ is unmarked and marks $y$; otherwise the oracle answers $y\suc_\el x$ and marks $x$. Notice that the oracle marks at most one unmarked candidate every time it is queried. We now claim that there must be at least $\nfrac{n}{10}$ votes which have been queried at least $\nfrac{m}{4}$ times. If not, then there exists $n-\nfrac{n}{10} = \nfrac{9n}{10}$ votes each of which has at least $m-\nfrac{m}{4} = \nfrac{3m}{4}$ candidates unmarked. In such a scenario, there exists a constant $N_0$ such that for every $m, n>N_0$, we have at least two candidates $x$ and $y$ who are unmarked in at least $(\lfloor\nfrac{n}{2}\rfloor + 1)$ votes each. Now if the algorithm outputs $x$, then we put $y$ at the first position in at least $(\lfloor\nfrac{n}{2}\rfloor + 1)$ votes and at the second position in the rest of the votes and this makes $y$ the (unique) Condorcet winner. If the algorithm does not output $x$, then we put $x$ at the first position in at least $(\lfloor\nfrac{n}{2}\rfloor + 1)$ votes and at the second position in the rest of the votes and this makes $x$ the (unique) Condorcet winner. Hence the algorithm fails to output correctly in both the cases contradicting the correctness of the algorithm. Also the resulting profile is single peaked on \TT with center at $y$ in the first case and at $x$ in the second case. Therefore the algorithm \AA must have query complexity $\Omega(mn)$.
\end{proof}

Our concluding result uses an intricate adversary argument, and shows that the query complexity for \CW for single peaked profiles in \Cref{thm:con_sp_ub} is essentially optimal, provided that the queries to different voters are not interleaved, as is the case with our algorithm.

\begin{theorem}\label{thm:con_sp_lb}
 Any \CW algorithm for single peaked profiles which does not interleave the queries to different voters has query complexity $\Omega(n\log m)$.
\end{theorem}

\begin{proof}
 Let a profile \PP be single peaked with respect to the ordering of the candidates $\suc = c_1\succ c_2\succ \cdots \succ c_m$. The oracle maintains two indices $\el$ and $r$ for every voter such that any candidate from $\{c_\el, c_{\el+1}, \ldots, c_r\}$ can be placed at the first position of the preference of the voter $v$ and still be consistent with all the answers provided by the oracle for $v$ till now and single peaked with respect to \suc. The algorithm initializes $\el$ to one and $r$ to $m$ for every voter. The oracle answers any query in such a way that maximizes the new value of $r-\el$. More specifically, suppose the oracle receives a query to compare candidates $c_i$ and $c_j$ with $i< j$ for a voter $v$. If the ordering between $c_i$ and $c_j$ follows, by applying transitivity, from the answers to the queries that have already been made so far for this voter, then the oracle answers accordingly. Otherwise the oracle answers as follows. If $i < \el$, then the oracle answers that $c_j$ is preferred over $c_i$; else if $j > r$, then the oracle answers that $c_i$ is preferred over $c_j$. Otherwise (that is when $\el \le i < j\le r$), if $j-\el > r-i$, then the oracle answers that $c_i$ is preferred over $c_j$ and changes $r$ to $j$; if $j-\el \le r-i$, then the oracle answers that $c_j$ is preferred over $c_i$ and changes $\el$ to $i$. Hence whenever the oracle answers a query for a voter $v$, the value of $r-\el$ for that voter $v$ decreases by a factor of at most two. Suppose the election instance has an odd number of voters. Let \VV be the set of voters. Now we claim that the first $\lfloor\nfrac{n}{5}\rfloor$ voters must be queried $(\log m - 1)$ times each. Suppose not, then consider the first voter $v^\pr$ that is queried less than $(\log m - 1)$ times. Then there exist at least two candidates $c_t$ and $c_{t+1}$ each of which can be placed at the first position of the vote $v^\pr$. The oracle fixes the candidates at the first positions of all the votes that have not been queried till $v^\pr$ is queried (and there are at least $\lceil\nfrac{4n}{5}\rceil$ such votes) in such a way that $\lfloor\nfrac{n}{2}\rfloor$ voters in $\VV\setminus\{v^\pr\}$ places some candidate in the left of $c_t$ at the first positions and $\lfloor\nfrac{n}{2}\rfloor$ voters in $\VV\setminus\{v^\pr\}$ places some candidate in the right of $c_{t+1}$. If the algorithm outputs $c_t$ as the Condorcet winner, then the oracle makes $c_{t+1}$ the (unique) Condorcet winner by placing $c_{t+1}$ at the top position of $v^\pr$, because $c_{t+1}$ is the unique median in this case. If the algorithm does not output $c_t$ as the Condorcet winner, then the oracle makes $c_t$ the (unique) Condorcet winner by placing $c_t$ at the top position of $v^\pr$, because $c_t$ is the unique median in this case. Hence the algorithm fails to output correctly in both the cases thereby contradicting the correctness of the algorithm.
\end{proof}

\longversion{
From \Cref{thm:cw_gen_ub,thm:con_sp_ub} we have the following result for any arbitrary tree.

\begin{theorem}\label{thm:cw_pc_ub}
 Let \TT be a tree with path cover number at most $k$. Then there is an algorithm for \CW for profiles which are single peaked on \TT with query complexity $\BigO(nk\log m)$.
\end{theorem}

\begin{proof}
 Let \PP be the input profile and $\QQ_i = (\XX_i, \EE_i) ~i\in[t]$ be $t(\le k)$ disjoint paths that cover the tree \TT. Here again, if the number of voters is even, then we remove any arbitrary voter and the algorithm outputs the Condorcet winner of the rest of the votes. The correctness of this step follows from the proof of \Cref{thm:cw_gen_ub}. Hence we assume, without loss of generality, that we have an odd number of voters. The algorithm proceeds in two stages. In the first stage, we find the Condorcet winner $w_i$ of the profile $\PP(\XX_i)$ for every $i\in[t]$ using \Cref{thm:con_sp_ub}. The query complexity of this stage is $\BigO(n\sum_{i\in[t]} \log|\XX_i|) = \BigO(nt\log (\nfrac{m}{t}))$. In the second stage, we find the Condorcet winner $w$ of the profile $\PP(\{w_i : i\in[t]\})$ using \Cref{thm:cw_gen_ub} and output $w$. The query complexity of the second stage is $\BigO(nt\log t)$. Hence the overall query complexity of the algorithm is $\BigO(nt\log(\nfrac{m}{t})) + \BigO(nt\log t) = \BigO(nk\log m)$.
\end{proof}
}


\longversion{
\section{Conclusions and Future Work}\label{sec:com}

We show algorithms for preference elicitation for profiles which are single peaked on trees. Moreover, we prove that the query complexity of our algorithms are optimal up to constant factors. We also show that we need to query fewer number of times than preference elicitation if we only want to find any weak Condorcet winner.

In this work we do not assume any partial information about the preferences. However, in many scenarios, we may have some knowledge about the profile. An interesting future direction of research is to study how having some partial information helps us reduce query complexity of preference elicitation.
}

\longversion{
\subsubsection*{Acknowledgement} Palash Dey wishes to gratefully acknowledge support from Google India for providing him with a special fellowship for carrying out his doctoral work.}

\longversion{\bibliographystyle{alpha}}
\shortversion{\bibliographystyle{named}}

\bibliography{Elicitation}


\end{document}